\documentclass[11pt]{article}
\usepackage[utf8]{inputenc}
\usepackage{fullpage}
\usepackage{amsmath,amsthm,amsfonts,amssymb}
\usepackage{tikz}
\usepackage{enumitem}
\usepackage{hyperref}
\usepackage[capitalize,nameinlink]{cleveref}
\usepackage{thmtools}
\usepackage{thm-restate}

\newcommand{\eps}{\varepsilon}
\newcommand{\N}{\mathbb{N}}

\newcommand{\F}{\mathbb{F}}

\newcommand{\seq}{\subseteq}

\newcommand{\poly}{\mathrm{poly}}

\newcommand{\ceil}[1]{{\lceil#1\rceil}}

\newtheorem{theorem}{Theorem}
\newtheorem{lemma}[theorem]{Lemma}
\newtheorem{claim}[theorem]{Claim}

\newtheorem{proposition}[theorem]{Proposition}

\title{Derandomization of Cell Sampling}
\author{
Alexander Golovnev\thanks{Georgetown University. Email: \texttt{alexgolovnev@gmail.com}.}
\and
Tom Gur\thanks{University of Warwick. Email: \texttt{tom.gur@warwick.ac.uk}. Tom Gur is supported by the UKRI Future Leaders Fellowship MR/S031545/1.}
\and
Igor Shinkar\thanks{Simon Fraser University. Email: \texttt{ishinkar@sfu.ca}.}
}
\date{}

\begin{document}

\maketitle
\begin{abstract}
Since 1989, the best known lower bound on static data structures was Siegel's classical cell sampling lower bound. Siegel showed an explicit problem with $n$ inputs and $m$~possible queries such that every data structure that answers queries by probing~$t$ memory cells requires space $s\geq\widetilde{\Omega}\left(n\cdot(\frac{m}{n})^{1/t}\right)$. In this work, we improve this bound for non-adaptive data structures to $s\geq\widetilde{\Omega}\left(n\cdot(\frac{m}{n})^{1/(t-1)}\right)$ for all $t \geq 2$.

For $t=2$, we give a lower bound of $s>m-o(m)$, improving on the bound $s>m/2$ recently proved by Viola over~$\F_2$ and Siegel's bound $s\geq\widetilde{\Omega}(\sqrt{mn})$ over other finite fields.
\end{abstract}

\section{Introduction}
For a finite field~$\F$, a \emph{static data structure problem} with $n$ inputs and $m$ possible queries is given by a function $f\colon\F^n\times[m]\to\F$. A non-adaptive static data structure (in the cell probe model) consists of two algorithms. The preprocessing algorithm takes an input $x\in\F^n$ and preprocesses it into $s$ memory cells $P\in\F^s$.
The query algorithm $Q$ takes an index $i\in[m]$, then non-adaptively probes at most $t$~memory cells from $P$, and has to compute $f(x,i)$.
That is, we require that $Q(P,i) = f(x,i)$ for all $i\in[m]$.
Here we assume that each input, memory cell, and query stores an element of the field~$\F$.\footnote{Typically in the cell probe model, the maintained memory is modelled by a sequence of
$s$ cells each holding a $w$-bit string. The parameter $w$ is called the \emph{word size} of the model.
In some scenarios the word size is equal to $\Theta(\log(n))$, however, $w$ is often considered a parameter (see e.g.,~\cite{Miltersen99}).
In this work, the word size corresponds to the number of bits required to represent an element of the field $\F$.
} We remark that in the cell probe model both the preprocessing and query algorithms are computationally unbounded.

Every data structure problem $f\colon\F^n\times[m]\to\F$ admits two trivial solutions:
\begin{itemize}
    \item $s=m$ and $t=1$, where in the preprocessing stage one precomputes the answers to all~$m$ queries. (This solution uses prohibitively large space.)
    \item $s=n$ and $t=n$, where one does not use preprocessing, but rather just stores the input. (This solution uses prohibitively large query time.)
\end{itemize}
A counting argument~\cite{Milt93} shows that a random data structure problem requires either almost trivial space $s\approx m$ or almost trivial query time $t\approx n$ (even in the case of adaptive data structures). The main challenge in this area is to prove a lower bound for an \emph{explicit} problem where each output can be computed in polynomial time.
The best known explicit lower bound was proven by Siegel~\cite{Siegel04} in 1989, and his technique was further developed in \cite{patrascu11structures, PTW10, Lar12}. This technique is now called cell sampling, and it will be discussed in greater detail later in this section. For an explicit problem, cell sampling gives us a lower bound of
\[
s\geq\widetilde{\Omega}\left(n\cdot\left(\frac{m}{n}\right)^{1/t}\right) \,,
\]
and this lower bounds holds even against adaptive data structures.
In particular, for $m=\poly(n)$, Siegel's result provides a problem that for linear space $s=O(n)$ requires logarithmic $t\geq\Omega(\log(n))$ query time. Alas, for super-linear space $s=n^{1+\eps}$, this best known lower bound only gives us the trivial $t\geq\Omega(1)$ bound. It is a major challenge in this area to improve on Siegel's bound.

While for the case of $t=1$, every non-trivial problem with $m$ queries requires space $s\geq m$,\footnote{For example, any problem where every pair of queries has at least $|\F|+1$ distinct pairs of values requires $s\geq m$ if $t=1$.} even the case of $t=2$ is not well understood. The cell sampling technique for $t=2$ gives a lower bound of $s\geq\widetilde{\Omega}(\sqrt{mn})$, but this is still far from the optimal bound of $s\geq \Omega(m)$ for $m=\poly(n)$. Only recently for the binary field $\F_2$, Viola~\cite{viola2019lower} proved a strong lower bound of $s\geq m/2$ on the space complexity for the case of $t=2$. Moreover, Viola~\cite{viola2019lower} showed that a better understanding of high lower bounds on the space complexity even for low values of $t$ will lead to resolving a~long-standing open problem in circuit complexity.

\paragraph{Our results.}
In this work, we further develop the cell sampling technique and improve its bound for the case of non-adaptive data structures~to
\[
s\geq\widetilde{\Omega}\left(n\cdot\left(\frac{m}{n}\right)^{1/(t-1)}\right) \,.
\]
On the one hand, this new bound does not improve asymptotic lower bounds on the query time~$t$ for any value of~$s$. On the other hand, for every fixed value of~$t$, the new bound gives an asymptotically stronger lower bound on~$s$. Furthermore, this bound essentially resolves the question for the case of $t=2$: for every field, and every number of queries $m=\poly(n)$, we give an explicit problem such that any data structure that probes $t=2$ memory cells requires memory $s\geq m-o(m)$ (see item~1 of \cref{thm:main_ds}). This improves on the bound of Viola~\cite{viola2019lower}, and answers a question asked by Rao and Natarajan Ramamoorthy~\cite{R20}.

\begin{restatable}{theorem}{mainthmds}\label{thm:main_ds}
Fix a finite field~$\F$ and a parameter $m=\poly(n)$.
\begin{enumerate}
  \item There exists an explicit problem with $n$ inputs and $m$ queries such that every non-adaptive static data structure solving it with query time~$t=2$ requires space $s \geq m - \tilde{O}(m/n)$.
  \item
    For every $t \geq 3$, there exists an explicit problem with $n$ inputs and $m$ queries such that every non-adaptive static data structure solving it with query time~$t$ requires space
    \begin{align*}
        s \geq \Omega\left( n\cdot\left(\frac{m}{n}\right)^{1/(t-1)}\cdot\frac{1}{2^t\log(n)\log(m)} \right) \,.
    \end{align*}
\end{enumerate}
\end{restatable}

The key step in the proof of \cref{thm:main_ds} is \cref{thm:main_graph} saying that every dense enough hypergraph contains a \emph{small} dense subgraph. 
Some of the techniques used in the proof of \cref{thm:main_graph} are well-known (for example, the proof of \cref{claim:bfs-tadpole} is similar to the standard upper bounds on the girth of a sparse graph~\cite[Chapter~IV, Theorem~1]{bollobas1998modern}). Nevertheless, we could not find results similar to \cref{thm:main_graph} in the literature (possibly due to the upper bound on the size of $S$ which may make this question less natural from the graph-theoretic point of view).

\begin{restatable}{theorem}{mainthm}\label{thm:main_graph}
Let $G = (V,E)$ be a multigraph with $|V|=s \geq 2$ vertices and $|E|=m\geq s(1+\eps)$ edges
for some $\eps=\eps(s)\in(0,1]$.
There exists a set of vertices $S \seq V$ of size $|S|\leq 8\log(s)\cdot\ceil{1/\eps}$ spanning at least $|S|+1$ edges.

Let $t \geq 3$ be an integer, and $G = (V,E)$ be a $t$-hypergraph with $|V|=s \geq 2$ vertices and $|E|=m$ hyperedges.
Let $k \in \N$ be a parameter such that $2^{t+2} \log(s)\leq k \leq s$.
If
\begin{align}\label{eq:magical-formula}
    m \geq 3s\left(\frac{2^{t+3}\cdot s \cdot \log(s)}{k}\right)^{t-2} \,,
\end{align}
then there exists a subset $S \seq V$ of size $|S| \leq k$ that spans at least $|S| + \frac{k}{2^{t+1} \log(s)}$ hyperedges.
\end{restatable}

In the statement of \cref{thm:main_graph} a \emph{multigraph} is a graph that may contain parallel edges and parallel self-loops, and a \emph{$t$-hypergraph} is a hypergraph where every edge has at most~$t$ vertices.
If all the vertices of an edge $e$ belong to a set of vertices $S$, then we say that $S$ spans the edge~$e$.

\paragraph{Comparison to the cell sampling bound.}
The classical cell sampling technique restricted to the case of non-adaptive data structures can be viewed as a slightly weaker version of \cref{thm:main_graph}. In the cell sampling argument, one picks $k$ \emph{random} vertices and proves that with non-zero probability they span at least $k+1$ hyperedges. This way, each $t$-hyperedge is spanned with probability $\approx\left(\frac{k}{s}\right)^t$, and the expected number of spanned edges is $\approx m\cdot\left(\frac{k}{s}\right)^t$. This leads to the lower bound of $m\gtrsim\Omega\left(s(\frac{s}{k})^{t-1}\right)$, which is weaker than the bound of $m\gtrsim\Omega(s\left(\frac{s}{k}\right)^{t-2})$ from \cref{thm:main_graph}. The contribution of this work is a deterministic way to choose~$k$ vertices that span $k+1$ hyperedges which improves on the aforementioned bound obtained by randomly sampling $k$ vertices.\footnote{We remark that the work~\cite{PTW10} also uses a deterministic process similar to cell sampling to prove the standard cell sampling lower bound even against adaptive data structures. In this work, we give a different deterministic process for a stronger bound against (weaker) non-adaptive data structures.}

The following proposition shows that the bound of \cref{thm:main_graph} is essentially tight, which poses a barrier on further improvements using this technique.
\begin{proposition}
Let $t\geq3$, and $G = (V,E)$ be a $t$-uniform hypergraph with $|V| = s$ vertices and $|E|=m$ edges sampled uniformly at random.
    Then with positive probability $G$ does not have a set of $k$ vertices spanning at least $k$ hyperedges for every $k\geq t$ satisfying $s \geq e^3 \cdot k\cdot \left(\frac{m}{k}\right)^{1/(t-1)}$.
\end{proposition}
\begin{proof}
    By the union bound over all $k$-subsets of vertices, and all $k$-subsets of edges, we have that
    the probability that $G$ has $k$ vertices spanning $\geq k$ hyperedges is at most
    \begin{equation*}
      \binom{s}{k} \cdot \binom{m}{k} \cdot \left(\frac{\binom{k}{t}}{\binom{s}{t}} \right)^k
      \,.
    \end{equation*}
    Using the inequalities $\left( \frac{a}{b}\right)^b \leq \binom{a}{b} \leq \left( \frac{a e}{b}\right)^b$,
    we have that this probability is bounded from above by
    \begin{equation*}
        \left( \frac{se}{k}\right)^k
        \cdot
        \left( \frac{me}{k}\right)^k
        \cdot
        \left(\frac{ke}{s}\right)^{tk}
        =
        \left( \frac{e^{t+2} m k^{t-2}}{s^{t-1}}\right)^k
        \leq \left( \frac{1}{e^{2t-5}}\right)^k
        \leq e^{-k}<1
      \,.
    \end{equation*}
\end{proof}
\section{Preliminaries}
All logarithms in this paper are to the base two. By $P_k$ and $C_k$ we denote the path and the cycle graphs on $k$ vertices, respectively. The length of a path is the number of edges in it.
By multigraphs we mean graphs that may contain parallel edges and (possibly parallel) self-loops. The degree of a vertex is the number of incident edges, and a self-loop adds two to the degree. For a multigraph $G=(V, E)$ and a subset of its vertices $S\subseteq V$, $G[S]$ denotes the subgraph of $G$ induced on the vertices $S$.

By \emph{$t$-hypergraphs} we mean hypergraphs where each edge contains \emph{at most} $t$ distinct vertices, and where parallel edges are allowed. A \emph{$t$-uniform hypergraph} is a $t$-hypergraph where each edge contains exactly $t$ vertices. We say that a set of vertices $S$ spans a hyperedge~$e$, if all the vertices of $e$ belong to $S$.

We will need the following auxiliary claim that shows that for every vertex~$v$ of a cubic graph, there is a small tadpole graph starting at~$v$.
\begin{claim}\label{claim:bfs-tadpole}
        Let $G = (V,E)$ be a multigraph with $|V|=s\geq 2$ vertices, and let $v \in V$.
        Suppose that $\deg(v) \geq 1$ and $\deg(u )\geq 3$ for all $u \in V \setminus \{v\}$.
        Then $G$ contains a path $(p_1,p_2,\dots,p_k)$ and a cycle $(c_1,c_2,\dots,c_\ell)$ such that $p_1 = v$, $p_k = c_1$, $k \geq 1, \ell \geq 1$, and $k+\ell \leq 4\log(s)$.
\end{claim}
In \cref{claim:bfs-tadpole} a cycle of length 1 means a self-loop, and a cycle of length 2 corresponds to a pair of parallel edges between two vertices.
In particular, if $v$ has a self-loop, then we can take the path to be $(v)$ and the cycle to be $(v)$.
\begin{proof}
Recall that a self-loop and a pair of parallel edges are cycles. For $s\leq 3$, a graph where all but one vertex have degree at least three contains a path from~$v$ of length $k\leq2$ to a cycle of length $\ell\leq2$. Therefore, $k+\ell\leq4\leq4\log(s)$, and the statement of the claim follows. Hence, in the following we assume that $s\geq 4$.

In order to prove the claim, we run the Breadth First Search (BFS) algorithm starting at the vertex~$v$. Let us say that $v$ is at the first level of the BFS tree, and let $t$ be the smallest integer such that the first $t$ levels of the BFS tree contain a cycle (note that a self-loop or a pair of parallel edges count as cycles). Such $t$ exists since $G$ has at most one vertex of degree 1.

Since all vertices in $V \setminus \{v\}$ have degrees at least three,
the number of vertices in the first $t-1$ levels of the BFS tree is at least $1+1+2+4+8+\dots+2^{t-3}=2^{t-2}$.
On the other hand, the total number of vertices is $|V|=s$. This implies that $s > 2^{t-2}$, and hence $t < \log(s)+2$.

Let $(c_1,c_2,\dots,c_\ell)$ be the obtained cycle.
Suppose without loss of generality that the level of $c_1$ in the BFS tree is the minimal among the $c_i$'s levels,
and let $(v=p_1,p_2,\dots,p_k = c_1)$ be the path from $v$ to $c_1$ in the BFS tree.
Observing that the length of the cycle from level $k$ to level $t$ has at most $\ell\leq 2(t-k)+1$ edges (where one edge may connect two vertices in the same level), we have that
$k+\ell \leq k + \ell + (k-1) = 2k+\ell-1 \leq 2t<2\log(s)+4\leq 4\log(s)$, as required.
\end{proof}

\section{Dense subgraphs in hypergraphs}
We are now ready to prove the first part of \cref{thm:main_graph}.
\begin{lemma}\label{lem:S+1-edges}
For every $\eps>0$ and every multigraph $G = (V,E)$ with $|V|=s \geq 2$ vertices and $|E|=m\geq s(1+\eps)$ edges, there exists a set of $k\leq 8\log(s)\cdot\ceil{1/\eps}$ vertices spanning at least $k+1$ edges.
\end{lemma}

\begin{proof}
    We repeatedly apply the following operations to the graph~$G$ as long as at least one of them is applicable.
    \begin{itemize}
        \item If $G$ contains an isolated vertex, then we remove this vertex.
        \item If $G$ contains a vertex of degree one, then we remove this vertex and the incident edge from the graph. In this case, we remove one edge and one vertex.
        \item If $G$ contains a vertex whose only incident edge is its self-loop, then we remove this vertex with the self-loop. Again, we remove one edge and one vertex.
        \item If $G$ contains a path of length $\ell\geq1/\eps+1$ consisting of vertices of degree two, then we remove all vertices of degree two (i.e., internal vertices) belonging to this path with all the incident edges. In this case, we remove $\ell-1\geq1/\eps$ vertices and $\ell$ edges.
    \end{itemize}
    Note that these four operations do not decrease the average degree of the graph, as the resulting graph has $s'$ vertices and $m'$ edges such that $m'/s'\geq m/s\geq(1+\eps)$. Each of the remaining vertices of degree two in the resulting graph belongs to a path of degree-two vertices of length less than $\ceil{1/\eps}$. We contract each such maximal path into an edge, and obtain a graph $G_1=(V_1,E_1)$ of minimum degree three. (Note that such contraction may create a self-loop, in case that the endpoints of the path are the same vertex.) We will used the following observation: since the length of each contracted path is at most $\ceil{1/\eps}$, when we expand a contracted edge back into a path, we add at most $\ceil{1/\eps-1}$ vertices to the graph.

    We apply \cref{claim:bfs-tadpole} to the graph $G_1$ and an arbitrary vertex $v_1 \in V_1$, and get a cycle $C^{(1)}_{\ell}$ in $G_1$ of length $\ell \leq 4\log(s)$. (Here, we ignore the path guaranteed by \cref{claim:bfs-tadpole}.)
    Next, we consider the following two cases.
    \begin{itemize}
      \item
        If $G_1[C^{(1)}_\ell]$ is a connected component in $G_1$, then since each vertex of~$G_1$ has degree at least~3,
        the vertices of $C^{(1)}_\ell$ must span at least $\ceil{1.5 \ell} \geq \ell+1$ edges. Let $S$ be the vertices of the subgraph $G_1[C^{(1)}_\ell]$, together with the vertices obtained by expanding $\ell+1$ of the contracted edges back into the vertices and edges of~$G$. Since each expanded edge adds the same number of vertices and edges, we have that $S$ spans at least $|S|+1$ edges. Since each of the  $\ell+1$ expanded edges introduces at most $\ceil{1/\eps-1}$ new vertices, we have that $|S|\leq \ell+ (\ell+1)\cdot\ceil{1/\eps-1} \leq 2\ell\cdot\ceil{1/\eps}\leq8\log(s)\cdot\ceil{1/\eps}$. Thus, the constructed set $S$ satisfies the required property.
      \item
        Otherwise, we contract $C^{(1)}_\ell$ into a new vertex $v_2$, and denote the obtained graph by $G_2$.
        Since $G_1[C^{(1)}_\ell]$ is not a connected component in $G_1$, it follows that $v_2$ is not an isolated vertex in $G_2$,
        and, hence, $\deg(v_2) \geq 1$.
        We now apply \cref{claim:bfs-tadpole} to $G_2$ and the vertex $v_2$, and get a path $P^{(2)}_{k'}$ and a cycle $C^{(2)}_{\ell'}$ in $G_2$ with $k'+\ell' \leq 4\log(s)$
        such that $v_2 \in P^{(2)}_{k'}$.

        Recall that the vertex $v_2$ in $G_2$ corresponds to $C^{(1)}_\ell$ in $G_1$. Thus, the set of vertices $S' = C^{(1)}_\ell \cup P^{(2)}_{k'} \cup C^{(2)}_{\ell'}$ forms two cycles connected by a path of length $k'-1$ in~$G_1$. Then $S'$ has $\ell + (k'-2) + \ell' \leq \ell+4\log(s)-1\leq 8\log(s)-1$ vertices, and spans at least $|S'| + 1$ edges in $G_1$. By expanding $|S'|+1$ contracted edges, we again have a set $S$ of vertices of~$G$ spanning at least $|S|+1$ edges. It remains to note that the number of vertices in~$S$ is $|S|\leq |S'|+(|S'|+1)\cdot\ceil{1/\eps-1}\leq (|S'|+1)\cdot\ceil{1/\eps}\leq8\log(s)\cdot\ceil{1/\eps}$.
    \end{itemize}
    This completes the proof of \cref{lem:S+1-edges}.
\end{proof}

\noindent The following lemma generalizes \cref{lem:S+1-edges} by finding a small induced subgraph with a large gap between the number of vertices and the number of edges.
\begin{lemma}\label{lem:gap-lemma}
Let $G = (V,E)$ be a multigraph with $|V| = s$ and $|E|=m$.
For every $g\geq 1$ satisfying $m \geq 2s+g+1$, there is a subset of vertices $S \seq V$ of size at most $|S| \leq 8g\log(s)$ that spans at least $|S|+g$ edges.
\end{lemma}

\begin{proof}
    The proof is by induction on $g \geq 1$.
    The base case of $g=1$ follows from \cref{lem:S+1-edges} with $\eps = 1$.
    Next we assume that the lemma is true for $g-1$, and prove it for $g \geq 2$.
    By the induction hypothesis there is a subset of vertices $S_{g-1} \seq V$ of size at most $|S| \leq 8(g-1)\log(s)$ such that $S_{g-1}$ spans at least $|S_{g-1}|+g-1$ edges.

    If $S_{g-1}$ spans $\geq |S_{g-1}|+g$ edges, then we are done. Otherwise, $S_{g-1}$ must span exactly $|S|+g-1$ edges.
        Consider a graph $G' =(V', E')$ obtained from $G$ by contracting $S_{g-1}$ into a new vertex $v^S$ and removing the edges with both ends in $S_{g-1}$.
        Then $G'$ has $s'=|V'| = s - |S| + 1$ vertices and $m'=|E'|= m-(|S|+g-1)$ edges.
        In particular, $m' \geq (2s + g + 1) - (|S|+g-1)  \geq 2s - 2|S| + 3= 2 s' + 1$.
        Therefore, by \cref{lem:S+1-edges},
        $G'$ has a subset of vertices $S' \seq V'$ of size at most $|S'| \leq 8\log(s)$ such that $S'$ spans at least $|S'|+1$ edges. We remark that $S'$ may or may not contain the vertex~$v^S$.

        By taking $S = S_{g-1} \cup (S'\setminus\{v^S\})$ we obtain a set of $|S|\leq8g\log(s)$ vertices spanning at least $|S|+g$ edges.
\end{proof}
\noindent We now finish the proof of \cref{thm:main_graph}.
\mainthm*
\begin{proof}
    The first part of the theorem is proven in \cref{lem:S+1-edges}. For the second part of the theorem, without loss of generality, we assume that $G$ is $t$-uniform. Indeed, if an edge of $G$ has fewer than $t$ vertices, then we extend this edge with arbitrary vertices, and the theorem statement for the new graph will imply the statement for~$G$.

    The proof of the second part of the theorem is by induction on $t \geq 2$.
    For the base case of $t=2$ the statement follows immediately from \cref{lem:gap-lemma}.
    Indeed, for $t=2$ the bound \eqref{eq:magical-formula} implies that $m \geq 3s$,
    and by \cref{lem:gap-lemma} with $g = \frac{k}{8 \log(s)}$ we get the desired conclusion.

    For the induction step, let us prove the statement of the theorem for $t$, assuming that it holds for $t-1$.
    Let $\ell$ be an integer such that $\frac{k}{2^{t+3} \log(s)} \leq \ell \leq \frac{k}{2^{t+2}\log(s)}$. Note that since $t\geq 3$ and $s\geq2$, we have that $\ell\leq k/4$. First we show that 
    there exists a subset $L \seq V$ of $\ell$ vertices such that the number of hyperedges touching them is
    $E_L = |\{e \in E : e \cap L \neq \emptyset\}| \geq \frac{\ell m}{s}$.
    Indeed, since $\sum_{v \in V} \deg(v) \geq tm$,
    there must exist a set $L \seq V$ such that $\sum_{v \in L} \deg(v) \geq \frac{\ell tm}{s}$.
    And since each hyperedge is counted in the sum at most $t$ times, it follows that the number of edges adjacent
    to $L$ is at least $\frac{\ell m}{s}$.

    Associate each hyperedge $e \in E_L$ with some vertex $v_e \in e \cap L$.
    That is, if $e$ contains a unique vertex $v_e$ in $e \cap L$, then we associate $e$ with this $v_e$,
    and if there is more than one such vertex, then we choose $v_e \in e \cap L$ arbitrarily.

    Define the graph $G^* = (V,E^*)$, where $E^* = \{e \setminus \{v_e\} : e \in E_L\}$.
    Note that $G^*$ has $s$ vertices, and by the assumption on~$m$ in~\eqref{eq:magical-formula}, the number of hyperedges (of size at most $t-1$) is at least
    \begin{align*}
        \frac{\ell m}{s}
        & \geq 3s\left(\frac{2^{t+3}\cdot s \cdot \log(s)}{k}\right)^{t-2} \cdot \frac{\ell}{s} \\
        & \geq 3s\left(\frac{2^{t+3}\cdot s \cdot \log(s)}{k}\right)^{t-2} \cdot \frac{k}{2^{t+3}\cdot s\cdot\log(s)} \\
        & = 3s\left(\frac{2^{t+2}\cdot s \cdot \log(s)}{k/2}\right)^{t-3}\\
        & \geq 3s\left(\frac{2^{t+2}\cdot s \cdot \log(s)}{k-\ell}\right)^{t-3}
        \enspace,
    \end{align*}
    where the last inequality uses $k-\ell \geq k/2$.
    Therefore, we can apply the induction hypothesis to the $(t-1)$-hypergraph $G^*$ with $k-\ell$ being the bound on the size of the guaranteed set.
    We get that $G^*$ has a subset $S^* \seq V$ of size $|S^*| \leq (k-\ell)$
    that spans at least $|S^*| + \frac{k-\ell}{2^{t}\log(s)}$ hyperedges.
    Define the set $S = S^* \cup L$.
    Therefore, $|S| \leq |S^*| + |L|  \leq (k-\ell) + \ell \leq k$,
    and since the number of hyperedges spanned by $S$ in $G$
    is at least the number of hyperedges spanned by $S^*$ in $G^*$,
    it follows that $S$ spans at least
    \begin{equation*}
        |S^*| + \frac{k-\ell}{2^{t}\log(s)}
        \geq |S| - \ell + \frac{k-\ell}{2^{t}\log(s)}
        \geq |S| - \frac{k}{2^{t+2}\log(s)} + \frac{3k}{2^{t+2}\log(s)}
        \geq |S| + \frac{k}{2^{t+1}\log(s)}
    \end{equation*}
    edges, as required.
\end{proof}

\section{Data structure lower bound}
We are now ready to prove the main theorem of this paper using \cref{thm:main_graph}. We will prove our data structure lower bound for $k$-wise independent functions. A function $f\colon\F^n\to\F^m$ is called $k$-wise independent if for every $k$-tuple $S$ of outputs, the uniform distribution of the $n$ inputs induces the uniform distribution on $S$. 

One way to construct $k$-wise independent functions utilizes linear error correcting codes. It is well known that the parity check matrix of a linear code with distance $k+1$ is $k$-wise independent. Therefore, one can define a $k$-wise independent data structure problem as the problem of multiplying an input vector $x\in\F^n$ by a fixed parity check matrix $M\in\F^{n\times m}$ of a code with a large distance. In particular, for fields of size $|\F|>m$, one can achieve $n$-wise independence by taking $M$ as the Vandermonde matrix. For smaller fields, one can take rate-optimal linear codes~\cite{macwilliams1977theory} and achieve $\Omega(n)$-wise independence for $m=O(n)$ and $\Omega(n/\log_{|\F|}(n))$-wise independence for every $m=\poly(n)$, which is tight~\cite{CGHFRS85}.

We remark that the result of \cref{thm:main_ds} applies to non-linear $k$-wise independent functions as well, and in fact it can be generalized to almost $k$-wise independent functions recovering the class of functions for which \cite{Siegel04} proved the cell sampling lower bound. For ease of exposition, in the proof below we show a lower bound for data structures computing $k=\Theta(n/\log(n))$-wise independent functions.

\mainthmds*
\begin{proof}
Consider a data structure for a $k=\Theta(n/\log(n))$-wise independent problem. For such a problem, in order to answer any $k$-tuple of queries, one needs to read at least $k$ memory cells. Indeed, every $k$-tuple of outputs of a $k$-wise independent function must take $|F|^{k}$ distinct values, and if it depended on $k-1$ memory cells it could only take at most $|F|^{k-1}$ distinct values.

For $t=2$, we construct a multigraph with $s$ vertices corresponding to the memory cells of the data structure, and $m$ edges, each corresponding to the pair of memory cells read for a query.
Let $k>8\log(s)$ and $\eps=\frac{16\log(s)}{k}$. If $m\geq s(1+\eps)$, then by the first part of \cref{thm:main_graph} we have a set of $k$ queries that depends on $k-1$ memory cells. Therefore, any data structure where in order to answer any $k$-tuple of queries, one needs to read at least $k$ memory cells, must satisfy $s\geq m/(1+\eps)\geq m(1-\eps)=m-\frac{16m\log(s)}{k}$. By plugging $k=\Theta(n/\log(n))$, we obtain the desired lower bound of $s\geq m-\widetilde{O}(m/n)$.

For $t\geq 3$, we construct a $t$-uniform hypergraph on $s$ vertices, where the vertices correspond to the memory cells of the data structure,
and $m$ hyperedges correspond to the $t$-tuples of memory cells read for each query. Let $k\geq 2^{t+2}\log(s)$. If $m\geq  3s\left(\frac{2^{t+3}\cdot s \cdot \log(s)}{k}\right)^{t-2}$, then 
by the second part of \cref{thm:main_graph}, there exists a set of $k+1$ queries that can be answered by $k$ memory cells. Therefore, every data structure that does not have such a $(k+1)$-tuple of queries must satisfy $s \geq \left(\frac{m}{3}\right)^{1/(t-1)}\cdot\left(\frac{k}{2^{t+3}\log(s)}\right)^{(t-2)/(t-1)}$. Setting $k=\Theta(n/\log(n))$ leads to the bound of $s\geq \Omega\left( n\cdot\left(\frac{m}{n}\right)^{1/(t-1)}\cdot\frac{1}{2^t\log(n)\log(m)} \right)$.
\end{proof}

\subsection*{Acknowledgments}
We would like to thank the anonymous reviewers whose detailed comments significantly helped us to improve the presentation of this result.
\bibliographystyle{alpha}
\bibliography{refs}
\end{document}